\documentclass[12pt]{article}
\newif\ifmobile
\mobilefalse 

\usepackage{geometry}
\usepackage{fancyhdr}
\usepackage{booktabs}
\usepackage{enumitem}
\usepackage{amsmath}
\usepackage{amssymb}
\usepackage{amsthm}
\usepackage{amsfonts}
\usepackage{slashed}
\usepackage{mathtools}      
\usepackage{algorithm}      
\usepackage{algpseudocode}  
\usepackage{authblk}        
\usepackage[unicode,psdextra,pdfencoding=auto]{hyperref}
\usepackage{doi}
\newcommand{\R}{\mathbb{R}}
\newcommand{\C}{\mathbb{C}}
\newcommand{\Hbb}{\mathbb{H}} 

\newcommand{\SU}[1]{\mathrm{SU}(#1)}
\newcommand{\U}[1]{\mathrm{U}(#1)}
\newcommand{\Spin}[1]{\mathrm{Spin}(#1)}

\let\Tr\relax 
\DeclareMathOperator{\Tr}{Tr}

\newcommand{\kappavec}{\vec{\kappa}}

\usepackage{amsthm} 

\theoremstyle{plain}
\newtheorem{theorem}{Theorem}[section]
\newtheorem{lemma}[theorem]{Lemma}
\newtheorem{proposition}[theorem]{Proposition}
\newtheorem{corollary}[theorem]{Corollary}

\theoremstyle{definition}
\newtheorem{definition}[theorem]{Definition}

\theoremstyle{remark}
\newtheorem{remark}[theorem]{Remark}

\ifmobile
\else
\fi
\title{Derivation of the Born Rule and Operational Quantum Formalism in the Accessibility Framework through Boundary Reduction}
\author[ ]{Everett Fall}
\author[1]{Hironori Kondo}
\affil[1]{Department of Material Science and Engineering, Massachusetts Institute of Technology, Cambridge, MA, 02139, USA}
\date{\today}
\begin{document}
\maketitle
\begin{abstract}
We show that the operational quantum formalism---the Born rule, L\"uders state updating, quantum interference, non-Markovian effective dynamics, and Bell inequality violation at the Tsirelson bound \(2\sqrt{2}\)---arises within Accessibility Theory (AT) from the Aperture construction together with explicit coherence and locality assumptions stated in the paper.  AT is a framework built on real graded spectral triples and a single algebraic selection principle.  The Principle of Universal Accessibility Balance requires three independent measures of the complexity of a spectral triple---its algebraic, gauge-theoretic, and geometric content---to be exactly equal and minimized, uniquely selecting the algebra \(\C \oplus \Hbb \oplus M_3(\C)\) and with it the Standard Model gauge group, particle content, four-dimensional Lorentzian spacetime, three generations, and gravitational dynamics.  Restriction to a codimension-one geometric boundary reduces this algebra to its commutative center \(\C \oplus \C \oplus \C\)---the Aperture---which defines a permanent information bottleneck for any embedded observer.  Coherence conditions on inference through this bottleneck, together with Gleason's theorem on the 48-dimensional internal Hilbert space, uniquely determine the Born rule; the remaining operational features follow from the same observer-level framework under the stated assumptions.  At the ontological level the theory is deterministic and state-realist, while the operational quantum formalism appears at the observer level as a consequence of structurally limited access to the underlying algebraic reality.
\end{abstract}
\section[Introduction]{Introduction}
\label{sec:introduction}

The two foundational frameworks of modern physics---the Standard Model and General Relativity---provide an extraordinarily accurate description of the physical world, yet they remain structurally disconnected. The Standard Model describes matter and its interactions through gauge symmetries but treats the gauge group \(\U{1} \times \SU{2} \times \SU{3}\), the number of fermion generations, and the spacetime dimension as empirical inputs. General Relativity describes the dynamics of spacetime geometry but says nothing about the origin of matter or the internal symmetries that govern it. Meanwhile, the operational quantum formalism---the observer-level probability rule, state-update prescription, and interference structure used throughout microscopic physics---is postulated rather than derived from a deeper geometric or algebraic principle.

This paper shows that the operational quantum formalism emerges from \textbf{Accessibility Theory} (AT), a framework built within the mathematics of real graded spectral triples \((\mathfrak{A}, \mathcal{H}, D)\) in the sense of Connes \cite{Connes1994, ConnesMarcolli2008}.  AT introduces a single selection principle---the \textbf{Principle of Universal Accessibility Balance}---which requires that three independent measures of the complexity of a spectral triple be exactly equal and minimized.  This principle determines the algebra, gauge group, particle content, spacetime dimension, and gravitational dynamics without free structural parameters \cite{AT}.  The present paper extends the framework to the observer level: we construct the algebraic interface through which an embedded observer accesses the underlying structure, and derive the quantum formalism as the unique inferential response to the resulting information bottleneck.

From this principle and the observer-level constraints it induces, the following deductive chain unfolds without free structural parameters:
\begin{enumerate}
    \item[\textbf{(i)}] The balance principle, applied to real graded spectral triples, uniquely selects the foundational algebra \(\mathfrak{A} = \C \oplus \Hbb \oplus M_3(\C)\)---precisely the algebra underlying the NCG formulation of the Standard Model \cite{ChamseddineConnesMarcolli2007}.
    \item[\textbf{(ii)}] From this algebra emerge the gauge group \(\U{1} \times \SU{2} \times \U{3}\), the Higgs field, and a 16-state fermionic generation with the correct Standard Model hypercharges.
    \item[\textbf{(iii)}] Quantum gauge anomaly cancellation excludes all spacetime dimensions \(n \geq 6\) and selects \(n = 4\) with Lorentzian signature.  Representational balance then yields exactly three fermion generations.
    \item[\textbf{(iv)}] The spectral action produces the Einstein--Hilbert action of General Relativity, and the quartically divergent vacuum energy contribution cancels at the balance equilibrium state.
    \item[\textbf{(v)}] Restriction to a codimension-one geometric boundary determines a boundary algebra \(\mathfrak{A}_b = \C \oplus \C \oplus \C\) and a structural information bottleneck---the \textbf{Aperture}---through which any embedded observer must interact with the underlying reality.
    \item[\textbf{(vi)}] Coherent valuation of Aperture outcomes forces probability on the Aperture record algebra; under a normalized finitely additive noncontextual extension to the full projection lattice, Gleason's theorem on the 48-dimensional Hilbert space yields the Born rule.  L\"uders state updating, quantum interference, and non-Markovian effective dynamics follow as consequences; Bell-inequality violation at the Tsirelson bound \(2\sqrt{2}\) follows in bipartite mode sectors under the standard local-QFT microcausality structure on the emergent Lorentzian manifold.
\end{enumerate}
The paper follows this chain.  Sections~\ref{sec:principles} and~\ref{sec:aperture} present the relevant structural results of \cite{AT} in condensed form: the balance principle and its consequences (i)--(iv), and the boundary reduction leading to the Aperture (v).  Section~\ref{sec:quantum} contains the main contribution of this paper: the emergence of the quantum formalism (vi).  Section~\ref{sec:discussion} summarizes the results and discusses open questions.  Detailed proofs of the structural results are given in \cite{AT}; derivations of the quantum formalism are contained in the appendices of this paper.

The central claim is that the quantum formalism is not an independent postulate but a derived consequence of the balance principle and the observer-level constraints induced by the Aperture construction.  A direct consequence of this framing is that the standard tension between unitary Schr\"odinger evolution and non-unitary collapse on measurement is resolved structurally: the ontological state evolves unitarily at all times, while what appears as collapse is an epistemic L\"uders update on the observer's density operator.

\paragraph{Relation to prior work.}
The Noncommutative Geometry Standard Model program \cite{ConnesChamseddine1996, ChamseddineConnesMarcolli2007} demonstrated that the Standard Model Lagrangian coupled to gravity can be recovered from the spectral action of an appropriately chosen spectral triple. That program, however, takes the algebra, spacetime dimension, and generation number as inputs. AT derives all three as outputs. The quantum reconstruction program \cite{Hardy2001, Chiribella2011, Dakic2011} derives the Hilbert space formalism from operational axioms, but assumes a quantum-mechanical framework at the outset. AT derives the quantum formalism from an algebraic and geometric starting point that is not \emph{a priori} quantum-mechanical. Recent work on stochastic reformulations of quantum theory \cite{Barandes2025} has identified intrinsic non-Markovian character as a hallmark of quantum dynamics. The AT framework provides a structural origin for this non-Markovianity: the coarse-graining imposed by the Aperture on deterministic bulk evolution produces effective stochastic dynamics that admit non-Markovian behavior under explicit algebraic conditions, offering a concrete algebraic mechanism for the phenomenon identified in \cite{Barandes2025}.
\section[The Principle of Universal Accessibility Balance]{The Principle of Universal Accessibility Balance}
\label{sec:principles}

This section summarizes the foundational principle of Accessibility Theory and its structural consequences, established in \cite{AT}.  We present the material needed to set up the Aperture construction and the quantum formalism derivation that follow in Sections~\ref{sec:aperture} and~\ref{sec:quantum}.  We begin by defining the three measures of accessibility that characterize a spectral triple, then state the balance and minimization conditions that constitute the principle.  Finally, we show how these conditions uniquely determine the algebra, gauge group, particle content, spacetime dimension, generation number, and gravitational dynamics.

\subsection[The Accessibility Measures]{The Accessibility Measures}
\label{subsec:measures}

A spectral triple \((\mathfrak{A}, \mathcal{H}, D)\) consists of a finite-dimensional real \(*\)-algebra \(\mathfrak{A}\), a complex Hilbert space \(\mathcal{H}\) carrying a faithful \(*\)-representation of \(\mathfrak{A}\), and a self-adjoint Dirac operator \(D\) \cite{Connes1994}.  To formulate the balance principle, we require a dimensionless measure that captures how much of a structure's internal complexity is accessible to interaction.  The goal is to define a common currency of information-carrying capacity for three sectors---the algebraic source, the gauge symmetries, and the spacetime geometry---so that we can require all three to be exactly capacity-matched: no sector may carry more or less accessible structure than the foundational algebra provides.  Simple dimension counting is insufficient for this purpose: for instance, the abelian group \(\U{1}^{\oplus 8}\) and the non-abelian group \(\SU{3}\) both have eight-dimensional Lie algebras, yet they describe very different physics because non-abelian generators participate in self-interactions.  The accessibility measures are designed to capture precisely such distinctions.

\paragraph{Foundational (Discrete) Accessibility.}
The algebraic core of the spectral triple has an intrinsic complexity: how large it is internally, and how richly it can act on a representation space.  The Foundational Accessibility captures both aspects.  The algebra is characterized by two quantities: its \emph{structural complexity} \(K = \dim_{\R}(\mathfrak{A})\), and its \emph{representational complexity} \(R = \sum_i \dim_{\C}(V_{\min,i})\), where \(V_{\min,i}\) is the minimal faithful complex \(*\)-representation of the \(i\)-th simple factor.  The Foundational Accessibility is defined as their product:
\begin{equation}
\label{eq:A_disc}
\mathcal{A}_{\mathrm{disc}}(\mathfrak{A}) \;\coloneqq\; K \cdot R\,.
\end{equation}
Here \(K\) counts the independent algebraic directions, while \(R\) counts the independent states on which those directions can act.  Their product therefore counts the effective number of algebra-on-representation pairings---the intrinsic measure of operator-space dimension for the algebraic source.  This value serves as the algebraic anchor to which both physical realizations must conform.

\paragraph{Accessibility of a realized sector.}
For a physical realization of the algebra, however, the raw dimensions of the sector are expressed in units specific to that domain (e.g., Lie algebra generators or spinor components).  To obtain a quantity comparable to the foundational anchor, these raw dimensions must first be converted to a common scale.  The procedure \cite{AT} is:
\begin{enumerate}
\item \emph{Normalize:} Determine the effective linear degrees of freedom \(d_{\mathrm{eff}}\) by dividing the raw dimensions of the sector by the dimension of its fundamental local unit---the smallest irreducible piece natural to that domain.
\item \emph{Square:} The accessibility is \(d_{\mathrm{eff}}^2\), reflecting that operator-space dimension scales quadratically with the effective linear dimension (the full matrix algebra on an \(n\)-dimensional complex space has dimension \(n^2\), not \(n\)).
\end{enumerate}

\paragraph{Symmetry Accessibility.}
The algebra must be realized physically through gauge symmetries, and not all symmetry generators carry the same dynamical weight: non-abelian generators participate in self-interactions while abelian ones do not.  The Symmetry Accessibility measures the effective dynamical capacity of the gauge sector, accounting for this distinction.  The algebra generates a unitary gauge group \(\mathcal{U}(\mathfrak{A})\) with Lie algebra \(\mathfrak{g} = \mathfrak{g}_{\mathrm{na}} \oplus \mathfrak{u}(1)^{\oplus A}\), where \(\mathfrak{g}_{\mathrm{na}}\) is the non-abelian (semi-simple) part of total dimension \(G\), and \(A\) counts the abelian factors.  In the canonically normalized basis, the trace-form metric on \(\mathfrak{g}\) assigns weight \(2\) to each non-abelian generator and weight \(1\) to each abelian generator \cite{Georgi1982, AT}.  Normalizing the total weight \(\Tr(g^{-1}) = 2G + A\) by the standard semi-simple unit (weight \(2\)) yields an effective linear dimension \(d_{\mathrm{eff}} = G + \tfrac{1}{2}A\).  The Symmetry Accessibility is
\begin{equation}
\label{eq:A_sym}
\mathcal{A}_{\mathrm{sym}}(\mathfrak{A}) \;\coloneqq\; \bigl(G + \tfrac{1}{2}A\bigr)^2\,.
\end{equation}

\paragraph{Continuous Accessibility.}
The geometric realization of the algebra lives on a spacetime manifold, and the relevant question is: how much of the Hilbert space is geometrically active, measured in units natural to the spacetime?  On an \(n\)-dimensional spacetime manifold, the fundamental local geometric unit is the spinor, whose dimension is \(2^{n/2}\) \cite{LawsonMichelsohn1989}.  Let \(H_{\mathrm{eff}} = \dim_{\C}((\ker D)^{\perp})\) be the dimension of the subspace on which the Dirac operator acts non-trivially.  Normalizing \(H_{\mathrm{eff}}\) by the spinor dimension and squaring gives
\begin{equation}
\label{eq:A_cont}
\mathcal{A}_{\mathrm{cont}}(n, H_{\mathrm{eff}}) \;\coloneqq\; \frac{H_{\mathrm{eff}}^2}{2^n}\,.
\end{equation}

This framework provides a natural unification of the two pillars of physics.  The Symmetry Accessibility governs the internal gauge sector, while the Continuous Accessibility governs the external geometric sector.  The balance principle asserts that both must exactly match the Foundational Accessibility of the algebraic source: the three sectors must be \emph{capacity-matched}.  The motivation is that the physical realizations must be faithful realizations of the algebraic source---complete and lossless in both directions.  If a realization carried \emph{less} accessibility than the algebra provides, some algebraic degrees of freedom would fail to participate in the physics: the realization would be an incomplete, lossy projection of the source.  If a realization carried \emph{more} accessibility than the algebra warrants, the physics would contain spurious degrees of freedom not grounded in the foundational structure.  Equality is the unique condition that excludes both pathologies: the total effective accessible capacity of each physical sector must match that of the algebraic source---no more, no less.

\subsection[The Balance Principle]{The Balance Principle}
\label{subsec:balance}

This requirement of exact accessibility-matching is the core of the theory.  The preceding section defined the three accessibility measures; we now formalize the requirement that they agree, and add the minimization condition that selects the unique algebraic solution.  Together, these two conditions constitute the Principle of Universal Accessibility Balance.

\paragraph{The Balance Condition.}
The physical realizations of the spectral triple must conform exactly to the foundational algebraic accessibility:
\begin{equation}
\label{eq:balance}
\mathcal{A}_{\mathrm{sym}} \;=\; \mathcal{A}_{\mathrm{disc}} \;=\; \mathcal{A}_{\mathrm{cont}}\,.
\end{equation}
Explicitly, for a spectral triple with algebraic parameters \((K, R)\), gauge parameters \((G, A)\), and geometric parameters \((H_{\mathrm{eff}}, n)\):
\begin{equation}
\label{eq:balance_explicit}
\bigl(G + \tfrac{1}{2}A\bigr)^2 \;=\; K \cdot R \;=\; \frac{H_{\mathrm{eff}}^2}{2^n}\,.
\end{equation}
The left equality constrains the algebra and its gauge group (\emph{Equation of Structural Balance}); the right equality constrains the Hilbert space and spacetime dimension (\emph{Equation of Representational Balance}).  Because \(K\), \(R\), \(G\), \(A\), \(H_{\mathrm{eff}}\), and \(n\) are all integers, these are Diophantine equations.

\paragraph{The Minimization Condition.}
Among all spectral triples satisfying the Balance Condition, the physical theory realizes the one with minimal Foundational Accessibility \(K \cdot R\).  This minimization operates at three levels:
\begin{itemize}
\item \emph{Algebraic} (permanent): Among all non-trivial, finite-dimensional, semi-simple, real \(*\)-algebras satisfying the Equation of Structural Balance, select the one with the smallest \(K \cdot R\).
\item \emph{Representational} (permanent): For the selected algebra and spacetime dimension, choose the Hilbert space representation with minimal total dimension \(H\) among all bimodule configurations satisfying the Equation of Representational Balance, faithfulness of the representation, and the existence of a chiral grading operator compatible with the spectral triple axioms.
\item \emph{Vacuum} (dynamical): The continuous parameters of the Dirac operator (masses, couplings) evolve toward the equilibrium configuration determined by this minimal accessibility value.
\end{itemize}

The first two levels fix all \emph{structural parameters}---the algebra, gauge group, spacetime dimension, Hilbert space dimension, and generation count---as quantized Diophantine invariants.  The third governs the approach of \emph{dynamical moduli} (the continuous parameters of \(D\)) toward equilibrium.  The balance principle thus draws a sharp distinction between what is structurally determined and what evolves.

\subsection[The Foundational Algebra]{The Foundational Algebra}
\label{subsec:algebra}

The first application of the principle is to determine the algebra.  The Equation of Structural Balance---the requirement that the Symmetry and Foundational Accessibilities are equal, \(\mathcal{A}_{\mathrm{sym}} = \mathcal{A}_{\mathrm{disc}}\)---is a Diophantine constraint: only finitely many algebras can satisfy it, and among those, the minimization condition selects a unique winner.  Remarkably, the algebra that emerges is precisely the one that Connes, Chamseddine, and Marcolli identified as the algebraic foundation of the Standard Model in the Noncommutative Geometry program \cite{ChamseddineConnesMarcolli2007}---but here it is derived from the balance principle rather than chosen to match observation.  The technical details of the search follow; the key result is Theorem~\ref{thm:algebra}.

Any candidate algebra \(\mathfrak{A}\) is a direct sum of simple real \(*\)-algebra factors of the form \(M_m(\Delta)\) with \(\Delta \in \{\R, \C, \Hbb\}\).  For each such factor, the quantities \(K\), \(R\), \(G\), and \(A\) can be computed from standard representation theory \cite{FultonHarris1991, Knapp2002}.  For a direct sum, these quantities are additive, so the Equation of Structural Balance~\eqref{eq:balance_explicit} becomes a Diophantine constraint on the list of simple factors.

\begin{theorem}[Uniqueness of the Foundational Algebra {\cite[Thm.~3.1]{AT}}]
\label{thm:algebra}
The Equation of Structural Balance~\eqref{eq:balance_explicit}, subject to the Minimization Condition, has a unique non-trivial solution among finite-dimensional semi-simple real \(*\)-algebras:
\begin{equation}
\label{eq:algebra}
\mathfrak{A} \;=\; \C \oplus \Hbb \oplus M_3(\C)\,,
\end{equation}
with \(K = 24\), \(R = 6\), \(G = 11\), \(A = 2\), and \(\mathcal{A}_{\mathrm{disc}} = K \cdot R = 144\).
\end{theorem}

The proof, given in full in \cite{AT}, proceeds by establishing analytical bounds that constrain the search space to a finite domain, then performing an exhaustive computational verification over all algebras within those bounds---more than 1.3 million candidate combinations.  The known solution \(\C \oplus \Hbb \oplus M_3(\C)\) has \(K \cdot R = 144\), which provides an upper bound: any candidate must satisfy \(K \cdot R \leq 144\).  Combined with the requirement \(K \geq n\) and \(R \geq n\) for an \(n\)-factor algebra, this caps the number of simple summands at \(n \leq 12\) and bounds the matrix dimension of each factor, rendering the search provably exhaustive.  A single structure satisfies the balance equation with the minimal accessibility value of \(144\).  One can verify directly: \((G + \frac{1}{2}A)^2 = (11 + 1)^2 = 144 = 24 \cdot 6 = K \cdot R\).

This algebra admits a concrete realization on a \(6\)-dimensional complex representation space \(\mathcal{H}_{\mathrm{rep}} = \C^1 \oplus \C^2 \oplus \C^3\).

\subsection[Emergent Gauge Structure and Particle Content]{Emergent Gauge Structure and Particle Content}
\label{subsec:particles}

The algebra determines both the gauge symmetries and the matter content.  The unitary group of \(\mathfrak{A}\) yields the Standard Model gauge group; the off-diagonal components of the Dirac operator produce the Higgs field; and the requirement of chiral balance on the bimodule representation singles out a 16-state fermionic generation whose hypercharges are fixed by anomaly cancellation.  None of these structures is introduced by hand---each is read off from the algebra derived in the preceding section.

\paragraph{Gauge group.}
The inner \(*\)-automorphisms of \(\mathfrak{A}\) are implemented by its unitary elements.  By the Skolem--Noether theorem, the automorphisms of each simple factor are inner \cite{Lam2001}, so the gauge group is obtained from the unitary groups of the simple summands, yielding
\begin{equation}
\label{eq:gauge_group}
\mathcal{U}(\mathfrak{A}) \;\cong\; \U{1} \times \SU{2} \times \U{3}\,.
\end{equation}
The \(\SU{2}\) factor arises because the unitary group of \(\Hbb\), viewed as a real \(*\)-algebra, is isomorphic to \(\Spin{3} \cong \SU{2}\) \cite{AT}.  The conventional Standard Model gauge group \(\U{1}_Y \times \SU{2}_L \times \SU{3}_c\) is recovered from this algebraic precursor through the usual unimodularity condition and hypercharge identification \cite{ChamseddineConnesMarcolli2007}.

\paragraph{Gauge and scalar fields.}
The Dirac operator \(D\) decomposes under the locus projections \(p_i\) of the algebra into components \(D_{ij} = p_i D p_j\).  Under a gauge transformation \(D \mapsto u D u^{-1}\), the diagonal components \(D_{ii}\) transform in the adjoint representation and correspond to \emph{vector gauge fields}, while the off-diagonal components \(D_{ij}\) (\(i \neq j\)) transform in bifundamental representations and correspond to \emph{scalar fields} \cite{Connes1996gravity}.  In particular, the component \(\Phi_{12}\) connecting the \(\C\) and \(\Hbb\) sectors transforms as an \(\SU{2}\) doublet charged under \(\U{1}\)---the Higgs field of the Standard Model \cite{vanSuijlekom2011}.

\paragraph{Fermionic spectrum.}
The Hilbert space \(\mathcal{H}\) decomposes as a direct sum of irreducible \(\mathfrak{A}\)-bimodules \(\mathcal{H}_{ij}\), with \(\dim(\mathcal{H}_{ij}) = d_i \, d_j\) and \((d_1, d_2, d_3) = (1, 2, 3)\) \cite{Connes1996gravity, vanSuijlekom2011, AT}.  The spectral triple axioms require a chiral grading that splits \(\mathcal{H} = \mathcal{H}_L \oplus \mathcal{H}_R\) with equal dimensions---a condition forced by invertibility of the Dirac operator.  Because the bimodule dimensions are small integers, this chiral balance creates a rigid Diophantine system with very few solutions.  The minimal non-trivial solution is a \textbf{16-state generation} with \(\dim(\mathcal{H}_L) = \dim(\mathcal{H}_R) = 8\) \cite{AT}.

Within this structure, the anomaly-cancellation constraints in four dimensions admit a unique real solution, fixing the hypercharges to their Standard Model values \cite{AT}.  The resulting spectrum is:
\begin{equation}
\label{eq:spectrum}
\begin{aligned}
&\underbrace{L_L = (\mathbf{1}, \mathbf{2})_{-1/2},\quad Q_L = (\mathbf{3}, \mathbf{2})_{+1/6}}_{\text{left-chiral}}\,,\\[6pt]
&\underbrace{\nu_R = (\mathbf{1}, \mathbf{1})_0,\quad e_R = (\mathbf{1}, \mathbf{1})_{-1},\quad u_R = (\mathbf{3}, \mathbf{1})_{+2/3},\quad d_R = (\mathbf{3}, \mathbf{1})_{-1/3}}_{\text{right-chiral}}\,,
\end{aligned}
\end{equation}
where quantum numbers are listed as \((\SU{3},\, \SU{2})_Y\).  This is precisely one generation of Standard Model fermions, including a right-handed neutrino.  The Higgs field \(\Phi_{12}\) arises as the off-diagonal component linking the \(\C\) and \(\Hbb\) sectors; within the AT construction, this coupling is required for the fermionic mass operator, as it is the only element in the algebra that maps \(\SU{2}\) doublets to singlets \cite{AT}.

\subsection[Spacetime Dimension and Generation Number]{Spacetime Dimension and Generation Number}
\label{subsec:spacetime}

The second application of the balance principle---the Equation of Representational Balance---determines the spacetime dimension and the number of generations.

The spectral triple axioms require the spacetime dimension \(n\) to be even (from the existence of the grading), and within the AT framework the relevant signature is shown to be Lorentzian \cite{AT}.  The balance equation \(H_{\mathrm{eff}}^2 / 2^n = 144\) admits solutions for each even \(n\), but the 16-state generation imposes an additional constraint.  Writing \(H = N_{\mathrm{gen}} \cdot 16\), the balance equation becomes \(N_{\mathrm{gen}}^2 = 9 \cdot 2^{n-4}\), which requires \(n - 4\) to be even.  This yields a family of candidate solutions at \(n = 4, 6, 8, \ldots\)\,.

\begin{theorem}[Exclusion of Higher Dimensions {\cite[Thm.~4.2]{AT}}]
\label{thm:dimension}
All even spacetime dimensions \(n \geq 6\) are excluded by quantum gauge anomaly cancellation.  Therefore \(n = 4\) is the unique consistent solution.
\end{theorem}

\begin{proof}[Proof sketch]
The anomaly polynomial in dimension \(n\) involves a pure abelian constraint of degree \(m = n/2 + 1\).  In six dimensions (\(m = 4\)), this quartic constraint reduces to \(Y_Q^2 = -5/36\), which has no real solution---the hypercharges derived from the bimodule structure are inconsistent with anomaly cancellation.  Higher even dimensions introduce additional anomaly conditions beyond those present in lower dimensions, so they are excluded \emph{a fortiori}.  The full proof is given in \cite{AT}.
\end{proof}

It follows that \(n = 4\), and substituting back gives \(N_{\mathrm{gen}} = 3\) and \(H = 48\):
\begin{equation}
\label{eq:dimension_generations}
n = 4\,, \qquad N_{\mathrm{gen}} = 3\,, \qquad H = 48\,.
\end{equation}
The spacetime dimension, the number of fermion generations, and the total Hilbert space dimension are all determined by the balance principle.

\subsection[Gravity and Vacuum Energy]{Gravity and Vacuum Energy}
\label{subsec:gravity}

The dynamics of the spectral triple are governed by the spectral action \(\mathcal{S}[D] = \Tr\bigl(f(D^2/\Lambda^2)\bigr)\) \cite{ChamseddineConnes1997}, whose asymptotic expansion via the heat kernel recovers the Standard Model Lagrangian coupled to gravity \cite{ChamseddineConnesMarcolli2007}.  In particular:

\paragraph{Einstein--Hilbert action.}
The asymptotic expansion of the spectral action via the heat kernel expresses \(\Tr\bigl(f(D^2/\Lambda^2)\bigr)\) as a sum of local geometric invariants weighted by powers of the cutoff \(\Lambda\) \cite{ChamseddineConnes1997}.  The second-order coefficient \(a_2\) contains a term proportional to the scalar curvature \(\mathcal{R}\) of the manifold, multiplied by the finite trace \(\Tr_F(\mathbf{I}) = H = 48\).  Upon integration, this generates the gravitational action
\begin{equation}
\label{eq:EH}
S_{\mathrm{grav}} \;=\; \frac{1}{16\pi G_N}\int_{\mathcal{M}} \mathcal{R}\,\sqrt{g}\;d^4x\,,
\end{equation}
with the Planck mass related to the cutoff scale and the Hilbert space dimension by \(M_P^2 \propto \Lambda^2 H\) \cite{ChamseddineConnes1997, ChamseddineConnesMarcolli2007, AT}.  General Relativity thus emerges from the same spectral triple that determines the gauge and matter sectors.

\paragraph{Vacuum energy cancellation.}
The spectral action also generates a field-independent vacuum energy term proportional to \(\Lambda^4\).  The balance principle identifies a distinguished reference configuration---the \textbf{Balance Equilibrium State}---at which the dynamical accessibility of the matter sector matches the foundational value \(K \cdot R = 144\).  Concretely, this equilibrium is the configuration of the Dirac operator at which the normalized spectral weight of the Yukawa couplings, \(\Sigma \coloneqq \frac{3}{4}\sum_f d_f\, y_f^2\), satisfies \(\Sigma = 1\)---a derived condition, not an imposed normalization \cite{AT}.  When the physical vacuum energy is defined relative to this equilibrium, the quartically divergent contribution cancels exactly, providing a structural resolution to the dominant term in the bare cosmological constant problem \cite{AT}.

\medskip

At this stage, the structural parameters of the theory are fully determined: the algebra \(\C \oplus \Hbb \oplus M_3(\C)\), the gauge group \(\U{1} \times \SU{2} \times \U{3}\), a four-dimensional Lorentzian spacetime, three generations of fermions with Standard Model quantum numbers, and gravitational dynamics governed by the Einstein--Hilbert action.  The open dynamical quantities---fermion masses, mixing angles, and the Majorana scale---are encoded in the Dirac operator and are not fixed at this structural stage.  We now turn to the consequences of the balance principle for the observer's interface with this structure.
\section[Boundary Reduction and the Aperture]{Boundary Reduction and the Aperture}
\label{sec:aperture}

The structural results of the previous section describe the universe as seen from outside---the full algebra, gauge group, and spacetime.  We now ask: what does this structure look like from the inside?  Said another way, we've established \emph{what exists} according to the model: a specific algebra, gauge group, particle content, and four-dimensional gravitational dynamics.  This section asks a different question: \emph{what can an embedded observer access?}  The boundary reduction is established in \cite{AT}; we present it here because it leads directly to the Aperture construction and the context-free record theorem, which provide the foundation for the quantum formalism derived in Section~\ref{sec:quantum}.

\subsection[The Complex Envelope and Center Reduction]{The Complex Envelope and Center Reduction}
\label{subsec:envelope}

The foundational algebra \(\mathfrak{A}\) is a \emph{real} \(*\)-algebra, but it acts on a \emph{complex} Hilbert space \(\mathcal{H}_F\) via a faithful \(*\)-representation \(\pi_F\).  To determine what an observer can access without choosing a preferred basis or orientation within the algebra, we first pass to the complex envelope---the smallest complex \(*\)-subalgebra of \(\mathcal{B}(\mathcal{H}_F)\) containing the image of \(\pi_F\)---and then extract its center.

The complex envelope is computed summand by summand.  The \(\C\) and \(M_3(\C)\) factors are already complex algebras, so they are unchanged: \(\pi_F(\C) + i\,\pi_F(\C) = \pi_F(\C)\) and likewise for \(M_3(\C)\).  The key step is the quaternionic factor.  In its minimal faithful complex representation, \(\Hbb\) embeds into \(M_2(\C)\) as a four-real-dimensional subalgebra.  Since \(\Hbb \otimes_{\R} \C \cong M_2(\C)\), the complex linear span of this image fills out all of \(M_2(\C)\).  The result is:

\begin{theorem}[Complex Envelope {\cite[Thm.~5.1]{AT}}]
\label{thm:envelope}
The complex envelope of \(\mathfrak{A}\) in representation satisfies
\begin{equation}
\label{eq:envelope}
\mathfrak{A}_F^{\C} \;\coloneqq\; \pi_F(\mathfrak{A}) + i\,\pi_F(\mathfrak{A}) \;\cong\; \C \oplus M_2(\C) \oplus M_3(\C)\,.
\end{equation}
\end{theorem}

The physical content of this result is that the quaternionic sector---which in the bulk describes \(\SU{2}\) gauge interactions---has its complexification yield the full matrix algebra \(M_2(\C)\) when represented on the complex Hilbert space.  The algebra has three simple summands, hence three minimal central projections \(\{P_1, P_2, P_3\}\) that decompose \(\mathcal{H}_F = \mathcal{H}_1 \oplus \mathcal{H}_2 \oplus \mathcal{H}_3\).

The \textbf{boundary algebra} is then defined as the center of the complex envelope:
\begin{equation}
\label{eq:boundary_algebra}
\mathfrak{A}_{F,b} \;\coloneqq\; Z\!\bigl(\mathfrak{A}_F^{\C}\bigr)\,.
\end{equation}
The center is singled out because it is precisely the subalgebra of observables (formalized in Definition~\ref{def:context_free} below) that commute with all internal unitaries---the physically unobservable relabelings within each sector.  Any larger commutative subalgebra (such as a maximal abelian subalgebra) would require selecting a preferred basis inside at least one non-abelian sector, breaking this invariance.  The center is therefore the unique commutative subalgebra invariant under all internal relabelings, and hence the unique basis-independent commutative algebra accessible to the observer, and the context-free record theorem in Section~\ref{subsec:observer} will confirm this at the record level: the observables invariant under all internal relabelings are exactly the central projections.

The center of a matrix algebra \(M_n(\C)\) consists of scalar multiples of the identity: \(Z(M_n(\C)) = \C \cdot \mathbf{I}_n\).  Since the center distributes over direct sums,
\begin{equation}
\label{eq:center_computation}
Z\!\bigl(\C \oplus M_2(\C) \oplus M_3(\C)\bigr) \;=\; \C \oplus \C \oplus \C\,.
\end{equation}
This is the boundary algebra: an abelian algebra isomorphic to three copies of \(\C\), with real dimension \(K_b = 3 \cdot 2 = 6\), representational complexity \(R_b = 3 \cdot 1 = 3\), and---importantly---exactly three minimal central projections corresponding to its three summands.

\subsection[Boundary Balance and the Resolution Ratio]{Boundary Balance and the Resolution Ratio}
\label{subsec:boundary_balance}

In the AT framework, the balance principle is applied equally to the induced codimension-one boundary geometry \cite{AT}.  For the boundary sector, the relevant parameters are the boundary algebra \(\mathfrak{A}_{F,b}\) with \((K_b, R_b) = (6, 3)\) and the boundary spacetime dimension \(n_b = n - 1 = 3\).  Applying the Equation of Representational Balance:
\begin{equation}
\label{eq:boundary_balance}
\frac{H_b^2}{2^{n_b}} = K_b \cdot R_b
\qquad\Longrightarrow\qquad
H_b^2 = 6 \cdot 3 \cdot 2^3 = 144
\qquad\Longrightarrow\qquad
H_b = 12\,.
\end{equation}
The boundary Hilbert space dimension is \(12\), compared to the bulk dimension \(H = 48\).  The \textbf{trace reduction ratio} is
\begin{equation}
\label{eq:xi}
\xi \;\coloneqq\; \frac{H_b}{H} = \frac{12}{48} = \frac{1}{4}\,.
\end{equation}

To quantify the information bottleneck, we use the \textbf{Resolution Ratio}: for any sector satisfying the balance equation, \(\mathcal{R} \coloneqq \mathcal{A}/K = R\), the representational complexity.  In the bulk, \(\mathcal{R}_{\mathcal{C}} = 6\); on the boundary, \(\mathcal{R}_{\mathcal{A}} = 3\).  Both exceed unity, meaning that in each sector the total accessible capacity outstrips what can be parametrized by the algebra's structural complexity alone.  In operational terms, an observer restricted to the boundary center cannot reconstruct the full state of the system from the data available through Aperture interactions: more operator-capacity sits behind each unit of explicit algebraic description than the observer can resolve.  This gap is permanent: the non-abelian character of the bulk algebra ensures that more degrees of freedom participate in the dynamics than can be resolved through the commutative boundary center.

\subsection[The Observer and the Three-Outcome Interface]{The Observer and the Three-Outcome Interface}
\label{subsec:observer}

The preceding subsections constructed the boundary algebra and quantified the information gap.  We now define the three central concepts---the Continuum (what exists), the Aperture (the interface), and the Observer (what is constrained to interact through the interface)---and prove that the observer's primitive measurement has exactly three outcomes.  This three-valued structure is not assumed; it is forced by requiring that observables be invariant under physically unobservable internal relabelings.  The limitation it imposes is the structural origin of the quantum formalism developed in the next section.

The \textbf{Continuum} is the full product spectral triple \((\mathfrak{A}_F \otimes C^{\infty}(\mathcal{M}),\; \mathcal{H}_F \otimes L^2(\mathcal{M}, S),\; D)\), carrying the total ontological state \(|\Psi\rangle\) that evolves unitarily.  The \textbf{Aperture} is the boundary algebra \(\mathfrak{A}_{F,b} \cong \C \oplus \C \oplus \C\): a commutative algebra of real dimension \(K_b = 6\), whose three minimal central projections define three primitive outcomes.  An \textbf{Observer} is any physical subsystem whose interaction with the Continuum is mediated exclusively through the Aperture.  The bulk algebra has structural complexity \(K = 24\); the boundary algebra has \(K_b = 6\)---a fourfold reduction.  Meanwhile, the resolution ratio drops from \(\mathcal{R}_{\mathcal{C}} = 6\) in the bulk to \(\mathcal{R}_{\mathcal{A}} = 3\) on the boundary, quantifying the permanent gap between what exists and what can be observed.

The question then becomes: what can the observer actually distinguish?  The answer is forced by an invariance argument.

\begin{definition}[Context-Free Record]
\label{def:context_free}
A \textbf{context-free record observable} is a projection-valued measure \(\{Q_r\}\) on \(\mathcal{H}_F\) that is invariant under all internal relabelings---unitaries that act independently within each sector \(\mathcal{H}_\alpha\) without mixing sectors.  A context-free record is \textbf{primitive} if it admits no strictly finer context-free refinement.
\end{definition}

\begin{theorem}[Context-Free Records Are Central]
\label{thm:central}
Every context-free record observable lies in the boundary center.  Each outcome projection has the form \(Q_r = \sum_{\alpha \in S_r} P_\alpha\) for some subset \(S_r \subseteq \{1, 2, 3\}\).  The unique primitive context-free record is \(\{P_1, P_2, P_3\}\): exactly three outcomes.
\end{theorem}

\begin{proof}
Decompose any candidate \(Q_r\) into blocks \(Q_r^{\alpha\beta} = P_\alpha\, Q_r\, P_\beta\).  Under a relabeling unitary \(U\) with sector-wise components \(U_\alpha \in \mathcal{U}(\mathcal{H}_\alpha)\), the block \(Q_r^{\alpha\beta}\) transforms as \(Q_r^{\alpha\beta} \mapsto U_\alpha\, Q_r^{\alpha\beta}\, U_\beta^{-1}\).  For \(\alpha \neq \beta\), choosing \(U_\alpha = e^{i\theta}\,\mathbf{I}_\alpha\) and \(U_\beta = \mathbf{I}_\beta\) gives \(e^{i\theta}\,Q_r^{\alpha\beta} = Q_r^{\alpha\beta}\) for all \(\theta\), hence \(Q_r^{\alpha\beta} = 0\).  For the diagonal blocks, invariance under all unitaries \(U_\alpha \in \mathcal{U}(\mathcal{H}_\alpha)\) forces \(Q_r^{\alpha\alpha}\) to lie in the commutant of \(\mathcal{U}(\mathcal{H}_\alpha)\), which is \(\C\,\mathbf{I}_\alpha\).  Since \(Q_r\) is a projection, \(Q_r^{\alpha\alpha} \in \{0, P_\alpha\}\).  Therefore \(Q_r = \sum_\alpha c_{r,\alpha}\, P_\alpha\) with each \(c_{r,\alpha} \in \{0, 1\}\), and the finest partition into such projections is \(\{P_1, P_2, P_3\}\).
\end{proof}

This result is the conceptual hinge of the theory.  The fact that the primitive observation has exactly three outcomes is not assumed---it is derived from the requirement that observables be invariant under physically unobservable internal relabelings.  Any finer resolution within a sector would require selecting a preferred basis, breaking the invariance.

The three outcomes correspond to the three summands of the complex envelope: the \(\C\) sector, the \(M_2(\C)\) sector, and the \(M_3(\C)\) sector.  The observer's entire context-free interaction with the Continuum is mediated through this three-valued interface.  Because the boundary algebra is commutative, the outcome algebra has a classical structure---each Aperture interaction yields a definite sector label.  The probabilistic character of quantum mechanics, as we show in the next section, arises not from any indeterminacy in the underlying state, but from the observer's structural inability to resolve the full non-abelian content of the algebra through this commutative bottleneck.
\section[Emergence of the Quantum Formalism]{Emergence of the Quantum Formalism}
\label{sec:quantum}

The previous section established that an observer embedded in the AT framework interacts with the Continuum through a commutative three-outcome interface---the Aperture.  We now show that the core features of the quantum formalism are derived from this structural constraint together with general coherence requirements on inference.  The Born rule, state updating, quantum interference, and non-Markovian effective dynamics are derived from this structural constraint together with explicit inferential commitments; Bell-inequality violation is recovered in bipartite mode sectors under the standard local-QFT microcausality structure on the emergent Lorentzian manifold.

\subsection[The Born Rule from Coherence Conditions]{The Born Rule from Coherence Conditions}
\label{subsec:born_rule}

For an observer who cannot resolve the full non-abelian content of the bulk algebra, the Aperture provides the only interface through which observational outcomes are recorded.  Such an observer may wish to reason inferentially under this structural incompleteness.  AT does not force the adoption of any inferential framework.  If the observer does adopt coherent numerical valuations of Aperture outcomes---linearity, positivity, and normalization on the space of gambles over the record algebra---then probability on the Aperture record algebra is forced (Appendix~\ref{app:born_rule}, Lemma~\ref{lem:aperture_valuation}).  If the observer extends this assignment to the full projection lattice as a normalized finitely additive noncontextual measure, Gleason's theorem forces the Born form.

The extended assignment is a map \(p : \mathcal{P}(\mathcal{H}_F) \to [0,1]\) satisfying three conditions:
\begin{itemize}
\item \emph{Normalization:} \(p(\mathbf{I}) = 1\).
\item \emph{Finite additivity:} For mutually orthogonal projections \(\{P_k\}\) summing to \(\mathbf{I}\), \(\,p\!\bigl(\sum_k P_k\bigr) = \sum_k p(P_k)\).
\item \emph{Noncontextuality:} \(p(P)\) depends only on the projection \(P\), not on the orthogonal decomposition of \(\mathbf{I}\) in which \(P\) appears.
\end{itemize}
Normalization and finite additivity on the Aperture record algebra follow from the valuation lemma; extending them to \(\mathcal{P}(\mathcal{H}_F)\) is an adopted enlargement of domain, not a theorem.  Noncontextuality is a substantive coherence commitment on its own; the Aperture context-freeness established in Theorem~\ref{thm:central} motivates,\footnote{If the primitive observable is context-free by construction, it is natural to require the same invariance of any probability assignment extending it.} but does not prove, it.  Lattice-level noncontextuality is not derived within AT, but adopted here as an additional coherence condition.  The rationale is that once projections are treated as events, the same projection should receive the same probability independently of the larger measurement context in which it appears.  Its role in Gleason-type derivations of the quantum probability rule is well established in the literature \cite{CavesFuchsManneRenes2004, Busch2003}.

The event structure on which the extended assignment operates is supplied by AT: the Hilbert space \(\mathcal{H}_F\) with \(\dim_{\C} \mathcal{H}_F = 48\), and its projection lattice \(\mathcal{P}(\mathcal{H}_F)\).  Since AT's underlying dynamics is unitary (Section~\ref{sec:principles}), sequential Aperture events naturally generate projections of the form \(U^\dagger P_\alpha U\);\footnote{Coherent reasoning about sequential outcomes is then naturally modeled by assignments on the projections generated by these orbits together with their orthogonal sums.} this makes the full projection lattice a natural and mathematically closed domain for such an extension.  With the three coherence conditions in place on \(\mathcal{P}(\mathcal{H}_F)\), together with \(\dim_{\C} \mathcal{H}_F = 48 \geq 3\), Gleason's theorem applies.

\begin{theorem}[Born Rule (Appendix~\ref{app:born_rule})]
\label{thm:born_rule}
Every probability assignment on \(\mathcal{P}(\mathcal{H}_F)\) (with \(\dim_{\C} \mathcal{H}_F = 48 \geq 3\)) obtained by coherent valuation of Aperture outcomes (Lemma~\ref{lem:aperture_valuation}) together with a normalized finitely additive noncontextual extension to the full projection lattice is uniquely of the form
\begin{equation}
\label{eq:born_rule}
p(P) = \Tr(\rho\, P)
\end{equation}
for a unique density operator \(\rho \geq 0\) with \(\Tr(\rho) = 1\).
\end{theorem}

This follows from Gleason's theorem \cite{Gleason1957}, which is a result of pure mathematics (frame-function analysis on \(\C^N\)) and does not make physical assumptions.  Its hypotheses are satisfied because the coherence conditions supply exactly the normalized finite additivity on orthogonal projections that Gleason requires, and the internal Hilbert space dimension \(48 \geq 3\) places us in the regime where the theorem applies.  The Born rule is therefore not a separately postulated quantum axiom.  It is the unique probability rule compatible with the event structure AT provides and the coherence conditions the observer adopts.

It is important to note what is and is not derived here.  The coherence conditions constitute a generalized probability calculus on projection events; they do not themselves assume the Born trace form \(p(P) = \Tr(\rho\,P)\).  What forces that specific quadratic form is Gleason's uniqueness theorem applied to the non-Boolean projection lattice of the AT Hilbert space---a lattice whose structure is determined by the algebra, not postulated.  On a classical Boolean event algebra, coherent probability assignments take many forms; on the projection lattice of a complex Hilbert space with \(\dim \geq 3\), Gleason's theorem collapses all such assignments to the Born rule.  The derivation therefore does not produce probability from nothing; it produces the uniquely quantum form of probability from coherent inference on the event structure forced by the Aperture and the AT Hilbert space.

Two further results follow from this foundation (Appendix~\ref{app:born_rule}).  The \textbf{L\"uders update rule}---the state \(\rho\) conditioned on a recorded boundary event \(\alpha\) becomes \(\rho_\alpha = P_\alpha\,\rho\,P_\alpha / \Tr(\rho\,P_\alpha)\)---is the unique posterior state satisfying the standard ideal-measurement axioms (outcome determinacy and preservation of sub-event probabilities).\footnote{The precise uniqueness statement is established in Appendix~\ref{app:born_rule}, Theorem~\ref{thm:born_luders}.} The \textbf{unitary evolution of epistemic states}---\(\rho \mapsto U\,\rho\,U^{\dagger}\) between interactions---is the unique update consistent with the ontological unitarity of the underlying dynamics (Section~\ref{sec:principles}).

\subsection[Quantum Interference]{Quantum Interference}
\label{subsec:interference}

Once the Born rule and unitary evolution of epistemic states are in hand, interference follows by a standard textbook derivation: transition probabilities do not compose classically.  We include the derivation to confirm that AT reproduces the phenomenon without additional postulates.

For an initial pure state \(|j\rangle\) evolving unitarily to time \(t\), the Born rule gives the transition probability to state \(|i\rangle\) as \(P(i,t \mid j,0) = |\langle i | U(t,0) | j \rangle|^2\).  Inserting a resolution of the identity at an intermediate time \(t_1\)---a purely mathematical rewriting, not an actual Aperture interaction at \(t_1\) (which would project the state and destroy the coherences described below)---yields
\begin{equation}
\label{eq:interference}
P(i, t \mid j, 0) = \biggl|\sum_k \mathcal{T}_{ik}(t, t_1)\,\mathcal{T}_{kj}(t_1, 0)\biggr|^2\,,
\end{equation}
where \(\mathcal{T}_{ij}(t_b, t_a) \coloneqq \langle i | U(t_b, t_a) | j \rangle\) is the transition amplitude from state \(|j\rangle\) at time \(t_a\) to state \(|i\rangle\) at time \(t_b\).  Expanding the squared modulus produces diagonal terms \(\sum_k |\mathcal{T}_{ik}(t,t_1)|^2\,|\mathcal{T}_{kj}(t_1,0)|^2\), which alone would satisfy the classical Chapman--Kolmogorov composition law, plus cross-terms that obstruct it. These cross-terms constitute quantum interference.  Since the unitary group on \(\mathcal{H}\) (with \(\dim \geq 2\)) contains elements that are not diagonal in any fixed basis, there exist configurations for which the cross-terms are nonvanishing and the Chapman--Kolmogorov rule fails.  In the Aperture framework, interference reflects the fact that unitary evolution preserves coherences between amplitudes that the observer's coarse-grained sector record cannot individually track.

\subsection[Non-Markovian Effective Dynamics]{Non-Markovian Effective Dynamics}
\label{subsec:nonmarkov}

A stronger signature of the Aperture's coarse-graining emerges when we consider \emph{sequential} Aperture interactions.  The sequence of recorded sector labels forms a classical stochastic process.  This process can be non-Markovian: two histories ending in the same recorded sector can nevertheless lead to different probabilities for the next sector label.  This occurs because the Aperture record need not determine the full post-interaction state within a sector, and later unitary evolution can make those unresolved intra-sector differences operationally visible.

To make this precise, fix an initial state \(\rho\) and a sequence of interaction times \(t_0 < t_1 < \cdots < t_n\).  At each time \(t_k\), an Aperture interaction records a sector label \(\alpha_k \in \{1, 2, 3\}\) and updates the epistemic state via the L\"uders map \(\mathcal{M}_{\alpha_k}(\sigma) = P_{\alpha_k}\,\sigma\,P_{\alpha_k} / \Tr(P_{\alpha_k}\,\sigma)\). Between interactions, the state evolves unitarily.  The resulting sequence \(\{\alpha_1, \ldots, \alpha_n\}\) is a classical stochastic process on the outcome space \(\{1, 2, 3\}\)---the \textbf{Aperture record process} (Appendix~\ref{app:nonmarkov}).

The mechanism is as follows.  After a L\"uders projection onto sector \(\alpha\), the post-interaction state \(\sigma_k\) is supported on the sector subspace \(P_\alpha\,\mathcal{H}_F\).  If this subspace has dimension \(r_\alpha \geq 2\), the intra-sector state retains information about earlier outcomes that the Aperture's three-valued record cannot capture.  Subsequent unitary evolution can then map this hidden intra-sector structure into observable differences in future transition probabilities.

More precisely, the transition probability from sector \(\alpha\) to sector \(\beta\) after one step of unitary evolution is \(\Tr(E_\beta^{(\alpha)}\,\sigma)\), where \(E_\beta^{(\alpha)} \coloneqq P_\alpha\,U^{\dagger}\,P_\beta\,U\,P_\alpha\) is the \textbf{transition effect operator} acting on \(P_\alpha\,\mathcal{H}_F\) (Appendix~\ref{app:nonmarkov}).  If \(E_\beta^{(\alpha)}\) is proportional to \(P_\alpha\), the transition probability is independent of \(\sigma\) and the process is locally Markovian.  Appendix~\ref{app:nonmarkov} establishes a conditional witness criterion: whenever \(E_\beta^{(\alpha)}\) is non-scalar on a sector with \(r_\alpha \geq 2\) and two histories produce distinguishable post-interaction states, the Aperture record is non-Markovian.  In the AT algebra, faithfulness of the representation guarantees that two of the three sectors have \(r_\alpha \geq 2\)---the \(M_2(\C)\) sector and the \(M_3(\C)\) sector---and the scalar-reset condition imposes strong algebraic constraints on \(U\).

The sector structure produces a characteristic \textbf{memory hierarchy}: one-dimensional sectors (\(r_\alpha = 1\)) are automatically Markovian, while sectors with larger internal dimension permit richer memory effects through the hidden intra-sector state.  This hierarchy is a direct operational signature of the three-sector structure derived in Theorem~\ref{thm:central}.

This result provides a concrete algebraic mechanism for the intrinsic non-Markovian character of quantum dynamics identified by Barandes \cite{Barandes2025} in the stochastic-quantum correspondence.  In that framework, quantum systems are equivalent to fundamentally non-Markovian stochastic processes, with the non-Markovianity arising from the structure of quantum amplitudes rather than from coupling to an external environment.  AT derives the same phenomenon from a specific structural source: the coarse-graining of deterministic unitary dynamics through the commutative Aperture, whose limited tomographic capacity cannot track the internal degrees of freedom that influence future boundary statistics.  Both frameworks thus identify non-Markovianity as a structural feature of quantum phenomena, and the transition effect operator \(E_\beta^{(\alpha)}\) provides a sharp algebraic criterion connecting the two perspectives.

A complementary structural parallel appears in the invasive-measurement reformulation of contextuality due to Navoni et al.\ \cite{Navoni2026}.  In that framework, contextual correlations require measurements that update the system in a way not determined by the outcome alone---formally, failure of the reset condition \(\Phi_M^o(s) = \omega_o\) on the post-measurement state.  The Aperture instrument \(\mathcal{M}_\alpha(\sigma) = P_\alpha\,\sigma\,P_\alpha / \Tr(P_\alpha\,\sigma)\) is precisely such an invasive update map, and the reset condition holds for it if and only if \(\mathrm{rank}(P_\alpha) = 1\): for rank-one \(P_\alpha = |\phi_\alpha\rangle\langle\phi_\alpha|\), one has \(\mathcal{M}_\alpha(\sigma) = |\phi_\alpha\rangle\langle\phi_\alpha|\) independent of \(\sigma\), while for \(r_\alpha \geq 2\), distinct states on \(P_\alpha\,\mathcal{H}_F\) remain distinguishable after projection.  Faithfulness of the representation guarantees that two of the three Aperture sectors have \(r_\alpha \geq 2\), placing AT in the invasive-non-reset regime that Navoni et al.\ identify as necessary for nonclassical sequential statistics.  The transition effect operator \(E_\beta^{(\alpha)}\) sharpens the abstract reset/non-reset dichotomy by detecting which intra-sector distinctions propagate to observable non-Markovian effects.

\subsection[Non-Local Realism and Bell Violation]{Non-Local Realism and Bell Violation}
\label{subsec:bell}

We now show that AT admits Bell-inequality violation without superluminal signaling, and classify the theory's interpretive character.

\begin{lemma}[Two-Qubit Mode Sector]
\label{lem:entangled_subspaces}
The internal lepton-doublet sector of the AT algebra contains a two-dimensional irreducible bimodule \(\mathcal{H}_{12} \cong V_1 \otimes V_2^* \cong \C^2\), with \(d_1 = 1\) and \(d_2 = 2\) \cite{AT}.  In its minimal faithful complex representation, the quaternionic factor \(\Hbb\) embeds into \(M_2(\C)\), and the complexified generators act as the Pauli matrices \(\sigma_x, \sigma_y, \sigma_z\) on this space.  Under the standard local-QFT structure on the emergent Lorentzian manifold---spatially localized mode excitations in disjoint regions, tensor-product composition of joint mode sectors, and Haag--Kastler microcausality---two spacelike-separated localized modes supported in disjoint regions \(A, B \subset \mathcal{M}\), each carrying this doublet degree of freedom, yield a composite mode sector
\begin{equation}
\mathcal{H}_{AB} \cong \C^2 \otimes \C^2,
\end{equation}
on which the quaternionic generators act independently as \(\sigma_{\hat{n}} \otimes \mathbf{I}\) and \(\mathbf{I} \otimes \sigma_{\hat{n}}\).
\end{lemma}

\begin{proof}
The product spectral triple \((\mathfrak{A}_F \otimes C^{\infty}(\mathcal{M}),\; \mathcal{H}_F \otimes L^2(\mathcal{M}, S),\; D)\) supports spatially localized mode excitations: for disjoint regions \(A, B \subset \mathcal{M}\), the internal doublet degree of freedom can be instantiated independently in each region, and the joint mode sector composes by tensor product, giving \(\mathcal{H}_{AB} \cong \C^2 \otimes \C^2\).  Kinematic commutativity \([X \otimes \mathbf{I},\, \mathbf{I} \otimes Y] = 0\) is the standard tensor-factor property.  Local commutativity of observable algebras at spacelike separation, \([\mathcal{A}(A),\, \mathcal{A}(B)] = 0\), is the standard Haag--Kastler microcausality condition of any relativistic local QFT built on the emergent Lorentzian manifold \cite{HaagKastler1964, Haag1992}, and we take it as given here.  The quaternionic \(\mathfrak{su}(2)\) generators act independently on each factor because each localized mode carries its own copy of the represented \(\Hbb\) action \cite{AT}.
\end{proof}

\begin{theorem}[CHSH Violation]
\label{thm:correlation_violation}
In the bipartite sector \(\mathcal{H}_{AB}\) of Lemma~\ref{lem:entangled_subspaces}, there exist local \(\pm 1\)-valued observables and a state such that, with probabilities computed via the Born rule (Theorem~\ref{thm:born_rule}), the CHSH expectation satisfies
\begin{equation}
\label{eq:chsh}
|\langle \mathsf{S} \rangle| = 2\sqrt{2}\,,
\end{equation}
violating the local-realist bound of \(2\) and saturating the Tsirelson bound \cite{Tsirelson1980}.
\end{theorem}

\begin{proof}
On each \(\C^2\) factor, define \(\sigma_{\hat{n}} = \hat{n} \cdot \vec{\sigma}\) for a unit vector \(\hat{n} \in \R^3\).  These have eigenvalues \(\pm 1\).  The singlet state is
\begin{equation}
|\Psi_s\rangle = \tfrac{1}{\sqrt{2}}\bigl(|{+}\rangle \otimes |{-}\rangle - |{-}\rangle \otimes |{+}\rangle\bigr).
\end{equation}
The total-spin-zero property \((\sigma_a \otimes \mathbf{I} + \mathbf{I} \otimes \sigma_a)|\Psi_s\rangle = 0\) for \(a = x, y, z\) implies \(\langle \Psi_s | \sigma_a \otimes \sigma_b | \Psi_s \rangle = -\delta_{ab}\).  By linearity, the correlation function is
\begin{equation}
C(\hat{n}_A, \hat{n}_B) \coloneqq \langle \Psi_s | \sigma_{\hat{n}_A} \otimes \sigma_{\hat{n}_B} | \Psi_s \rangle = -\hat{n}_A \cdot \hat{n}_B.
\end{equation}
Define the CHSH operator \(\mathsf{S} = A \otimes B - A \otimes B' + A' \otimes B + A' \otimes B'\) with \(A = \sigma_{\hat{n}_A}\), \(A' = \sigma_{\hat{n}_{A'}}\), \(B = \sigma_{\hat{n}_B}\), \(B' = \sigma_{\hat{n}_{B'}}\).  Choose settings in a common plane with relative angles \(\theta_{AB} = \theta_{A'B} = \theta_{A'B'} = \pi/4\) and \(\theta_{AB'} = 3\pi/4\).  Then
\begin{equation}
\langle \mathsf{S} \rangle = -\cos(\pi/4) + \cos(3\pi/4) - \cos(\pi/4) - \cos(\pi/4) = -2\sqrt{2},
\end{equation}
so \(|\langle \mathsf{S} \rangle| = 2\sqrt{2}\).
\end{proof}

\begin{lemma}[No-Signaling]
\label{lem:no_signalling}
In the bipartite mode decomposition \(\mathcal{H}_{AB} \cong \mathcal{H}_A \otimes \mathcal{H}_B\), any completely positive trace-preserving (CPTP) map \(\Phi_A\) applied locally to mode \(A\) leaves the reduced state on \(B\) unchanged.
\end{lemma}

\begin{proof}
Let \(\Phi_A(X) = \sum_\alpha K_\alpha X K_\alpha^\dagger\) with \(\sum_\alpha K_\alpha^\dagger K_\alpha = \mathbf{I}_A\); then \(\Phi_A \otimes \mathrm{id}_B\) has Kraus operators \(K_\alpha \otimes \mathbf{I}\).  For any state \(\rho\) on \(\mathcal{H}_{AB}\),
\begin{align}
\Tr_A\bigl((\Phi_A \otimes \mathrm{id}_B)(\rho)\bigr)
&= \sum_\alpha \Tr_A\bigl((K_\alpha \otimes \mathbf{I})\,\rho\,(K_\alpha^\dagger \otimes \mathbf{I})\bigr) \nonumber\\
&= \sum_\alpha \Tr_A\bigl((K_\alpha^\dagger K_\alpha \otimes \mathbf{I})\,\rho\bigr) \nonumber\\
&= \Tr_A(\rho),
\end{align}
where the second equality uses the identity \(\Tr_A\bigl((M_A \otimes \mathbf{I})\,\rho\,(N_A \otimes \mathbf{I})\bigr) = \Tr_A\bigl((N_A M_A \otimes \mathbf{I})\,\rho\bigr)\) for any \(M_A, N_A \in \mathcal{B}(\mathcal{H}_A)\), which follows by expanding \(\rho\) in a product basis and applying cyclicity of the \(A\)-trace in the \(A\)-factor (partial traces are not cyclic in general, but this tensor-factor identity holds because the \(\mathbf{I}_B\) factors commute through the partial trace over \(A\)).  The final equality then uses \(\sum_\alpha K_\alpha^\dagger K_\alpha = \mathbf{I}_A\).  Therefore the marginal statistics in region \(B\) are independent of any local operation in region \(A\).
\end{proof}

\begin{remark}[Origin of nonseparability]
\label{rem:prespatial_nonlocality}
The Bell violation does not arise from superluminal causal influence but from the nonseparability of the global state \(|\Psi\rangle\) relative to the emergent subsystem decomposition.  At the algebraic level, the spectral triple \((\mathfrak{A}, \mathcal{H}, D)\) defines a single global Hilbert space without reference to spatial localization.  The manifold \(\mathcal{M}\) emerges from the spectral data of \(D\) \cite{AT}, and the spectral action generates local dynamical equations, but the resulting states \(|\Psi\rangle \in \mathcal{H}\) need not be separable with respect to the emergent spatial subsystem decomposition.  The correlations reflect the nonseparability of a single global state when decomposed into the spatial subsystem structure of the emergent Continuum.
\end{remark}

\paragraph{Classification.}
Accessibility Theory describes a world in which the ontological state \(|\Psi\rangle\) exists as a definite element of the Hilbert space and evolves deterministically under unitary evolution at all times.  What the embedded observer encounters is not this state but its image through the Aperture---a commutative three-outcome interface that cannot resolve the non-abelian content of the bulk algebra.  Every distinctive feature of the quantum formalism developed above traces to that single bottleneck.  The Born rule is the unique coherent probability assignment compatible with the projection lattice behind the interface.  Interference is the persistence in the observer's inferences of amplitude coherences that the sector record cannot individually resolve.  Non-Markovian memory is hidden intra-sector structure leaking into future transition statistics through unitary evolution.  Bell-violating correlations arise from the nonseparability of the global state when that state is decomposed into emergent spatial modes; no-signaling is simply the tensor-product locality of that decomposition.

In ontological-model terms \cite{HarriganSpekkens2010}, AT is realist and deterministic at the level of the underlying global state, while the operational quantum state---the density operator \(\rho\)---functions epistemically for the embedded observer.  The inaccessible information is not carried by Bell-local hidden variables in the sense of a pre-assigned outcome map \(\lambda \in \Lambda\); it resides in the global algebraic degrees of freedom that cannot be recovered from the Aperture record, a gap whose magnitude is fixed by the resolution ratio \(\mathcal{R}_{\mathcal{C}} = 6 > 1\).  Accessibility Theory is thus realist and deterministic at the ontological level, non-local in the sense of Bell-inequality violation, and epistemic in its use of the quantum state---with the epistemic limitation grounded in a concrete, quantifiable algebraic bottleneck rather than in any subjective or operational convention.
\section[Discussion and Conclusion]{Discussion and Conclusion}
\label{sec:discussion}

\subsection[Non-abelian Gauge Content]{Non-abelian Gauge Content}

One structural feature worth noting is how AT treats non-abelian gauge content.  Because the boundary algebra \(\mathfrak{A}_{F,b}\) is the center of the complex envelope of \(\mathfrak{A}_F\), the primitive context-free record resolves only the central sector label---not within-sector directions of \(M_2(\C)\) or \(M_3(\C)\).  Non-abelian bulk dynamics can nevertheless influence observed record statistics through the transition effect operator \(E_\beta^{(\alpha)} = P_\alpha U^\dagger P_\beta U P_\alpha\) (Appendix~\ref{app:nonmarkov}), whose dependence on the full unitary \(U\) encodes non-abelian couplings in the bulk Hamiltonian.  This structure is compatible with the standard gauge-theoretic principle that only gauge-invariant quantities are directly observable \cite{Weinberg1995, Strocchi2013}, but AT's primitive cut is algebraically sharper than gauge invariance: the center \(\C \oplus \C \oplus \C\) is a three-dimensional commutative algebra, whereas the full gauge-invariant operator algebra of a non-abelian theory is much richer.

\subsection[The Measurement Problem]{The Measurement Problem}

The standard quantum measurement problem arises from the apparent tension between unitary Schr\"odinger evolution and a non-unitary collapse postulate on measurement.  Within AT, this tension is resolved by distinguishing the ontological state \(|\Psi\rangle\), which evolves unitarily at all times, from the epistemic density operator \(\rho\), which updates by L\"uders conditioning upon record of a sector outcome.  What appears to the embedded observer as collapse is this L\"uders update (Theorem~\ref{thm:born_luders})---a change in the observer's inferential description, not a physical event on the ontological state.  The ontological and epistemic layers thus have different update rules without inconsistency, because they describe different objects.

\subsection[Emergence of the Classical Layer]{Emergence of the Classical Layer}

In decoherence-based accounts of quantum foundations, the classical character of macroscopic observables emerges dynamically from environmental entanglement: a system's reduced density matrix becomes approximately diagonal in a preferred basis through interaction with the environment \cite{Zurek2003, Schlosshauer2005}.  Within AT, the boundary algebra \(\mathfrak{A}_{F,b} = \C \oplus \C \oplus \C\) is literally commutative---an algebraic consequence of the boundary reduction, not a dynamical approximation.  This provides a structural source of classicalization at the boundary interface itself, complementary to decoherence processes that remain relevant for the emergence of classical-looking behavior within the bulk.

\subsection[Bulk-to-boundary State Map]{Bulk-to-boundary State Map}

The bulk-to-boundary relationship is structurally established through the boundary algebra \(\mathfrak{A}_{F,b}\) and its dimension \(H_b = 12\), but the explicit state-space map from the bulk ontological state to the boundary epistemic state has not yet been constructed.  The trace reduction ratio \(\xi = 1/4\) quantifies the dimensional relationship; formalizing the corresponding channel on states is a natural direction for future work.

\subsection[Concluding Remarks]{Concluding Remarks}

Within the Accessibility Theory framework, the Standard Model, General Relativity, and the operational quantum formalism trace back to a single balance principle rather than being introduced as independent foundational postulates.  The Aperture construction provides the structural setting in which the characteristic features of the operational quantum formalism---probabilistic outcomes, interference, non-Markovian dynamics, and non-local correlations---arise for an embedded observer whose access to a deterministic algebraic reality is mediated through a permanent information bottleneck.

\newpage
\appendix

\section{Foundational Proofs and Derivations}
\subsection[Structural Derivation of the Born Rule]{Structural Derivation of the Born Rule}
\label{app:born_rule}

This appendix derives the Born rule, the Lüders update rule, and the unitary evolution of epistemic states as the unique inferential framework compatible with the Aperture interface and the inferential inputs listed below. The derivation builds on the operator-algebraic structure already established in Accessibility Theory together with the following inferential inputs:

\begin{enumerate}
    \item[(i)] coherent numerical valuation of Aperture outcomes (linearity, positivity, and normalization on the space of gambles over the record algebra),
    \item[(ii)] a normalized finitely additive noncontextual extension of the resulting probability assignment to the full projection lattice \(\mathcal{P}(\mathcal{H}_F)\),
    \item[(iii)] finite-dimensionality of the internal Hilbert space, with \(\dim_{\C}\mathcal{H}_F = 48 \geq 3\),
    \item[(iv)] ordinary Bayesian conditioning on recorded boundary events,
    \item[(v)] unitary ontological evolution of the global state,
    \item[(vi)] finite-dimensional trace duality.
\end{enumerate}

Throughout, let \(\mathcal{H}\) denote the finite-dimensional complex Hilbert space of the internal sector under consideration, with \(\dim_{\C}(\mathcal{H}) = n\). In AT, \(n = 48\) for the full internal Hilbert space \(\mathcal{H}_F\). Let \(\mathcal{B}(\mathcal{H}) \cong M_n(\C)\) and let \(\mathcal{P}(\mathcal{H})\) denote its projection lattice.

\subsubsection[Probability on the Aperture Record Algebra]{Probability on the Aperture Record Algebra}

The Aperture interface (Section~\ref{sec:aperture}) provides the observer with three mutually exclusive primitive outcomes, labeled by the minimal central projections \(\{P_1, P_2, P_3\}\) of the boundary algebra \(\mathfrak{A}_{F,b} \cong \C \oplus \C \oplus \C\).  The record algebra \(\mathcal{R}_{\text{Ap}}\) is the finite Boolean algebra generated by these projections:
\begin{equation}
\mathcal{R}_{\text{Ap}} = \bigl\{0,\; P_1,\; P_2,\; P_3,\; P_1 + P_2,\; P_1 + P_3,\; P_2 + P_3,\; \mathbf{I}\bigr\}.
\end{equation}
An observer reasoning inferentially about Aperture outcomes may assign numerical valuations to gambles over these outcomes.  A \textbf{gamble} is a real-valued function \(f : \{1, 2, 3\} \to \R\) specifying a payoff \(f(\alpha)\) for each primitive outcome.  The space of gambles is the three-dimensional real vector space \(\mathcal{F} = \R^{\{1, 2, 3\}}\), with basis \(\{e_1, e_2, e_3\}\) defined by \(e_\alpha(\beta) = \delta_{\alpha\beta}\), and constant gamble \(\mathbf{1}(\alpha) = 1\) for all \(\alpha\).

\begin{lemma}[Valuation Representation on the Aperture Record Algebra]
\label{lem:aperture_valuation}
Let \(V : \mathcal{F} \to \R\) be a valuation functional satisfying:
\begin{enumerate}
    \item[(V1)] \emph{Linearity:} \(V(af + bg) = a\,V(f) + b\,V(g)\) for all \(f, g \in \mathcal{F}\) and \(a, b \in \R\),
    \item[(V2)] \emph{Positivity:} if \(f(\alpha) \geq 0\) for all \(\alpha\), then \(V(f) \geq 0\),
    \item[(V3)] \emph{Normalization:} \(V(\mathbf{1}) = 1\).
\end{enumerate}
Then there exists a unique probability vector \((p_1, p_2, p_3)\) with \(p_\alpha \geq 0\) and \(\sum_{\alpha} p_\alpha = 1\) such that
\begin{equation}
V(f) = \sum_{\alpha = 1}^{3} p_\alpha\, f(\alpha) \qquad \text{for all } f \in \mathcal{F}.
\end{equation}
Consequently, the induced assignment \(p : \mathcal{R}_{\text{Ap}} \to [0, 1]\), defined on record events \(E = \sum_{\alpha \in S} P_\alpha\) by \(p(E) = \sum_{\alpha \in S} p_\alpha\), is a finitely additive probability measure on \(\mathcal{R}_{\text{Ap}}\).
\end{lemma}

\begin{proof}
Every gamble \(f \in \mathcal{F}\) decomposes uniquely as \(f = \sum_{\alpha = 1}^{3} f(\alpha)\, e_\alpha\).  By (V1),
\begin{equation}
V(f) = \sum_{\alpha = 1}^{3} f(\alpha)\, V(e_\alpha).
\end{equation}
Define \(p_\alpha \coloneqq V(e_\alpha)\).  Since \(e_\alpha(\beta) \geq 0\) for all \(\beta\), (V2) gives \(p_\alpha \geq 0\).  From \(\mathbf{1} = e_1 + e_2 + e_3\) and (V1), (V3),
\begin{equation}
1 = V(\mathbf{1}) = V(e_1) + V(e_2) + V(e_3) = p_1 + p_2 + p_3.
\end{equation}
Thus \((p_1, p_2, p_3)\) is a probability vector and \(V(f) = \sum_\alpha p_\alpha\, f(\alpha)\) as claimed.  Uniqueness follows from the uniqueness of the basis decomposition: any other probability vector representing \(V\) would have to satisfy \(p'_\alpha = V(e_\alpha) = p_\alpha\).

For the induced measure, let \(E = \sum_{\alpha \in S} P_\alpha\) and \(E' = \sum_{\alpha \in S'} P_\alpha\) be record events with \(S \cap S' = \varnothing\).  Then \(E + E' = \sum_{\alpha \in S \cup S'} P_\alpha\), and
\begin{equation}
p(E + E') = \sum_{\alpha \in S \cup S'} p_\alpha = \sum_{\alpha \in S} p_\alpha + \sum_{\alpha \in S'} p_\alpha = p(E) + p(E').
\end{equation}
Normalization \(p(\mathbf{I}) = p_1 + p_2 + p_3 = 1\) is immediate.
\end{proof}

\begin{remark}[Operational content of the valuation axioms]
\label{rem:valuation_operational}
The three axioms (V1)--(V3) express minimal coherence requirements on numerical valuations of outcome gambles.  Linearity expresses additive aggregation: the valuation of a combined gamble is the sum of the valuations of its parts, and scaling a gamble by a constant scales its valuation by the same constant.  Positivity expresses monotonicity under sure dominance: a gamble that never pays negatively cannot be valued negatively.  Normalization fixes the unit of valuation by assigning the value \(1\) to the constant unit gamble.  These axioms do not assume probability; they derive it.  The Lemma shows that any valuation functional obeying them is necessarily an expectation with respect to a unique probability vector over the three Aperture outcomes.
\end{remark}

Lemma~\ref{lem:aperture_valuation} forces a probability assignment on the Aperture record algebra \(\mathcal{R}_{\text{Ap}}\) once the observer adopts coherent valuation.  The remainder of the appendix considers a normalized finitely additive noncontextual extension of this assignment to the full projection lattice \(\mathcal{P}(\mathcal{H}_F)\) and applies Gleason's theorem to derive the Born form.

\subsubsection[Gleason-Type Reconstruction on the AT Hilbert Space]{Gleason-Type Reconstruction on the AT Hilbert Space}

\begin{theorem}[Gleason's theorem, finite-dimensional complex form {\cite{Gleason1957}}]
\label{thm:gleason_external}
Let \(\mathcal{H}\) be a complex Hilbert space with \(\dim_{\C}\mathcal{H} \geq 3\). Let
\[
\mu: \mathcal{P}(\mathcal{H}) \to [0,1]
\]
satisfy:
\begin{enumerate}
    \item[(a)] \(\mu(\mathbf{I}) = 1\),
    \item[(b)] for every finite family of pairwise orthogonal projections \(\{P_k\}\) with \(\sum_k P_k = \mathbf{I}\),
    \[
    \mu\!\left(\sum_k P_k\right) = \sum_k \mu(P_k).
    \]
\end{enumerate}
Then there exists a unique positive operator \(\rho \in \mathcal{B}(\mathcal{H})\) with \(\Tr(\rho) = 1\) such that
\[
\mu(P) = \Tr(\rho\,P)
\qquad\text{for all } P \in \mathcal{P}(\mathcal{H}).
\]
\end{theorem}

\begin{remark}[Applicability inside AT]
\label{rem:gleason_applicability}
Gleason's theorem is a result of pure mathematics (frame-function analysis on the unit sphere of \(\C^N\)) that makes no physical assumptions. Its hypotheses are satisfied in AT: the internal Hilbert space has \(\dim_{\C}\mathcal{H}_F = 48 \geq 3\) (Finite Hilbert Space Constraint Corollary of \cite{AT}), and a normalized finitely additive noncontextual extension of the Aperture-level measure of Lemma~\ref{lem:aperture_valuation} to \(\mathcal{P}(\mathcal{H}_F)\) supplies exactly the normalized finite additivity on orthogonal projections required by conditions (a)--(b).

For individual sectors \(\mathcal{H}_\alpha = P_\alpha\,\mathcal{H}_F\), the probability assignment is defined on the full projection lattice of \(\mathcal{H}_F\), and Gleason's theorem applied globally---using \(\dim_{\C}\mathcal{H}_F = 48 \geq 3\)---yields a density operator \(\rho\) whose restriction to each sector automatically provides the Born rule form.  The global application therefore suffices; no separate sector-level hypothesis is required.
\end{remark}

\begin{theorem}[Born Rule for Projections]
\label{thm:born_rule_appendix}
Every probability assignment on \(\mathcal{P}(\mathcal{H}_F)\) obtained by coherent valuation of Aperture outcomes (Lemma~\ref{lem:aperture_valuation}) together with a normalized finitely additive noncontextual extension to the full projection lattice is of the form
\begin{equation}
p(P) = \Tr(\rho\,P)
\end{equation}
for a unique density operator \(\rho \geq 0\) with \(\Tr(\rho) = 1\).
\end{theorem}

\begin{proof}
By Lemma~\ref{lem:aperture_valuation}, coherent valuation of Aperture outcomes induces a finitely additive probability measure on the Aperture record algebra \(\mathcal{R}_{\text{Ap}}\).  By the extension hypothesis, this measure extends to an assignment \(p: \mathcal{P}(\mathcal{H}_F) \to [0,1]\) satisfying:
\begin{itemize}
    \item \(p(\mathbf{I}) = 1\) (normalization),
    \item \(p\!\left(\sum_k P_k\right) = \sum_k p(P_k)\) for pairwise orthogonal \(\{P_k\}\) summing to \(\mathbf{I}\) (finite additivity on the full lattice),
    \item \(p(P)\) depends only on \(P\) (noncontextuality),
    \item \(\dim_{\C}(\mathcal{H}_F) = 48 \geq 3\) (Finite Hilbert Space Constraint Corollary of \cite{AT}).
\end{itemize}
These are precisely the hypotheses of Theorem~\ref{thm:gleason_external}. Applying Gleason's theorem yields a unique density operator \(\rho \in \mathcal{B}(\mathcal{H}_F)\) with \(\rho \geq 0\), \(\Tr(\rho) = 1\), and \(p(P) = \Tr(\rho\,P)\) for all \(P \in \mathcal{P}(\mathcal{H}_F)\).
\end{proof}

\begin{corollary}[Canonical extension to effects]
\label{cor:born_effect_extension}
The projection-level probability assignment of Theorem~\ref{thm:born_rule_appendix} admits a unique canonical extension to the effect algebra
\[
\mathcal{E}(\mathcal{H}) \coloneqq \{E \in \mathcal{B}(\mathcal{H}) \mid 0 \leq E \leq \mathbf{I}\},
\]
given by
\begin{equation}
p(E) \coloneqq \Tr(\rho\,E).
\end{equation}
In particular, if \(E, F \in \mathcal{E}(\mathcal{H})\) satisfy \(E + F \leq \mathbf{I}\), then \(p(E + F) = p(E) + p(F)\).
\end{corollary}

\begin{proof}
By Theorem~\ref{thm:born_rule_appendix}, the projection assignment is represented by a unique density operator \(\rho\). Define \(p(E) \coloneqq \Tr(\rho\,E)\) for \(E \in \mathcal{E}(\mathcal{H})\). This agrees with the original assignment on projections. If \(E + F \leq \mathbf{I}\), then
\[
p(E + F) = \Tr(\rho(E + F)) = \Tr(\rho\,E) + \Tr(\rho\,F) = p(E) + p(F)
\]
by linearity of the trace. The extension is canonical because \(\rho\) is uniquely determined by the projection-level assignment.
\end{proof}

\subsubsection[Uniqueness of Lüders Conditioning]{Uniqueness of Lüders Conditioning}

\begin{theorem}[Lüders Update Rule]
\label{thm:born_luders}
Let \(P_\alpha\) be a recorded boundary event with \(\Tr(\rho\,P_\alpha) > 0\). Among all posterior density operators \(\rho_\alpha\), there is a unique one satisfying:
\begin{enumerate}
    \item[(L1)] \(\Tr(\rho_\alpha\,P_\alpha) = 1\),
    \item[(L2)] \(\Tr(\rho_\alpha\,Q) = \dfrac{\Tr(\rho\,Q)}{\Tr(\rho\,P_\alpha)}\) for every projection \(Q \leq P_\alpha\).
\end{enumerate}
It is
\begin{equation}
\rho_\alpha = \frac{P_\alpha\,\rho\,P_\alpha}{\Tr(\rho\,P_\alpha)}.
\end{equation}
\end{theorem}

\begin{proof}
Set \(\widetilde{\rho}_\alpha \coloneqq P_\alpha\,\rho\,P_\alpha / \Tr(\rho\,P_\alpha)\).

\medskip\noindent\textbf{Verification of (L1).}\quad
\(\Tr(\widetilde{\rho}_\alpha\,P_\alpha) = \Tr(P_\alpha\,\rho\,P_\alpha\,P_\alpha) / \Tr(\rho\,P_\alpha) = \Tr(P_\alpha\,\rho\,P_\alpha) / \Tr(\rho\,P_\alpha) = 1\).

\medskip\noindent\textbf{Verification of (L2).}\quad
Let \(Q \leq P_\alpha\). Then \(Q = P_\alpha\,Q\,P_\alpha\), hence by cyclicity of the trace,
\[
\Tr(\widetilde{\rho}_\alpha\,Q)
= \frac{\Tr(P_\alpha\,\rho\,P_\alpha\,Q)}{\Tr(\rho\,P_\alpha)}
= \frac{\Tr(\rho\,P_\alpha\,Q\,P_\alpha)}{\Tr(\rho\,P_\alpha)}
= \frac{\Tr(\rho\,Q)}{\Tr(\rho\,P_\alpha)}.
\]

\medskip\noindent\textbf{Uniqueness.}\quad
Let \(\sigma\) be any density operator satisfying (L1) and (L2). Since \(\Tr(\sigma\,P_\alpha) = 1\) and \(\sigma \geq 0\), we have \(\sigma = P_\alpha\,\sigma\,P_\alpha\), so \(\sigma\) is supported on \(P_\alpha\,\mathcal{H}\). For every unit vector \(|\psi\rangle \in P_\alpha\,\mathcal{H}\), the rank-one projection \(|\psi\rangle\langle\psi| \leq P_\alpha\) satisfies
\[
\langle\psi|\sigma|\psi\rangle = \Tr(\sigma\,|\psi\rangle\langle\psi|) = \Tr(\widetilde{\rho}_\alpha\,|\psi\rangle\langle\psi|) = \langle\psi|\widetilde{\rho}_\alpha|\psi\rangle
\]
by (L2). Hence the quadratic forms of \(\sigma\) and \(\widetilde{\rho}_\alpha\) agree on all vectors in \(P_\alpha\,\mathcal{H}\). Since both operators are self-adjoint, the polarization identity
\[
\langle\phi|A|\psi\rangle = \tfrac{1}{4}\sum_{k=0}^{3} i^k \langle\phi + i^k\psi|A|\phi + i^k\psi\rangle
\]
implies \(\langle\phi|\sigma|\psi\rangle = \langle\phi|\widetilde{\rho}_\alpha|\psi\rangle\) for all \(|\phi\rangle, |\psi\rangle \in P_\alpha\,\mathcal{H}\). Therefore \(\sigma = \widetilde{\rho}_\alpha\).
\end{proof}

\subsubsection[Compatibility with Unitary Evolution]{Compatibility with Unitary Evolution}

\begin{theorem}[Epistemic Evolution]
\label{thm:born_evolution}
Let the ontological state evolve unitarily by \(U\). Then the unique evolution law on epistemic states compatible with the Born rule at all times is
\begin{equation}
\rho \longmapsto U\,\rho\,U^\dagger.
\end{equation}
\end{theorem}

\begin{proof}
Let \(P\) be any projection. Consistency of epistemic probabilities with ontological unitary evolution requires \(p_t(P) = p_0(U^\dagger P\,U)\). Using the Born rule at time \(0\),
\[
p_t(P) = \Tr(\rho_0\,U^\dagger P\,U) = \Tr(U\,\rho_0\,U^\dagger\,P).
\]
Therefore the time-\(t\) probability assignment is represented by \(\rho_t \coloneqq U\,\rho_0\,U^\dagger\). Uniqueness follows from uniqueness of the density-operator representation in Theorem~\ref{thm:born_rule_appendix}.
\end{proof}

\subsubsection[Summary]{Summary}

\begin{theorem}[Uniqueness of the AT Inferential Framework]
\label{thm:born_full}
Within Accessibility Theory, the unique inferential framework compatible with (i)~coherent valuation of Aperture outcomes (Lemma~\ref{lem:aperture_valuation}), (ii)~a normalized finitely additive noncontextual extension of the resulting probability assignment to the full projection lattice, (iii)~finite-dimensionality of the internal Hilbert space with \(\dim \geq 3\), (iv)~Bayesian conditioning on boundary events, and (v)~unitary ontological evolution consists of the following three rules. The \textbf{Born rule} assigns probabilities to projections via a unique density operator:
\begin{equation}
p(P) = \Tr(\rho\,P). \label{eq:born_summary}
\end{equation}
Between interactions, the epistemic state evolves by \textbf{unitary conjugation}:
\begin{equation}
\rho \;\longmapsto\; U\,\rho\,U^\dagger. \label{eq:unitary_summary}
\end{equation}
Upon recording a boundary event \(\alpha\), the epistemic state updates by \textbf{Lüders conditioning}:
\begin{equation}
\rho \;\longmapsto\; \frac{P_\alpha\,\rho\,P_\alpha}{\Tr(\rho\,P_\alpha)}. \label{eq:luders_summary}
\end{equation}
\end{theorem}

\begin{proof}
Born rule: Theorem~\ref{thm:born_rule_appendix}. Effect extension: Corollary~\ref{cor:born_effect_extension}. Lüders update: Theorem~\ref{thm:born_luders}. Unitary evolution: Theorem~\ref{thm:born_evolution}. Uniqueness at each step is established in the respective proofs.
\end{proof}

\begin{remark}[Relationship to Gleason's theorem]
\label{rem:born_gleason}
Gleason's theorem \cite{Gleason1957} is used in Theorem~\ref{thm:born_rule_appendix} to establish that the extended probability assignment is uniquely representable by a density operator. The present derivation builds on this to determine the full inferential framework: the canonical extension to effects, the update rule, and the evolution law. These go beyond the scope of Gleason's theorem, which provides only the Born rule form on projections.
\end{remark}

\begin{remark}[Extended interactions and arbitrary states]
\label{rem:born_extended}
The Born rule form, the Lüders update, and the unitary evolution law are valid for any density operator \(\rho\): they are the derived inferential rules, not properties of any specific state. When an observer acquires information through extended interactions that resolve intra-sector degrees of freedom, the effective \(\rho\) within a sector can be any density operator compatible with the acquired data.
\end{remark}
\subsection[Non-Markovianity of the Aperture Record]{Non-Markovianity of the Aperture Record}
\label{app:nonmarkov}

\begin{definition}[Aperture Interaction Instrument]
\label{def:aperture_instrument}
Let \(\{P_\alpha\}_{\alpha = 1}^{3} \subset \pi(\mathfrak{A}_{F,b})\) be the minimal central projections of the boundary center, corresponding to the three simple summands of \(\mathfrak{A}_F^{\C} \cong \C \oplus M_2(\C) \oplus M_3(\C)\).  An \textbf{Aperture interaction} at time \(t_k\) is the event of registering a sector label \(\alpha_k \in \{1, 2, 3\}\) together with the associated conditioning map on epistemic states
\begin{equation}
\mathcal{M}_\alpha(\sigma) \coloneqq \frac{P_\alpha\,\sigma\,P_\alpha}{\Tr(P_\alpha\,\sigma)},
\end{equation}
defined whenever \(\Tr(P_\alpha\,\sigma) > 0\).  Between interactions, epistemic states evolve by the unitary channel
\begin{equation}
\mathcal{U}_{k,k-1}(\sigma) \coloneqq U(t_k, t_{k-1})\,\sigma\,U(t_k, t_{k-1})^\dagger.
\end{equation}
The L\"uders form of the conditioning map is the unique posterior state satisfying the standard ideal-measurement axioms (outcome determinacy and preservation of sub-event probabilities) (Theorem~\ref{thm:born_luders}).
\end{definition}

\begin{definition}[Aperture Record Process]
\label{def:aperture_record}
Fix an initial epistemic state \(\rho\) and a sequence of interaction times \(t_0 < t_1 < \cdots < t_n\).  The \textbf{Aperture record process} is the classical stochastic process \(\{\alpha_k\}_{k=1}^{n}\) with outcome space \(\{1, 2, 3\}\) and joint distribution
\begin{equation}
\mathrm{Pr}(\alpha_n, \ldots, \alpha_1) \coloneqq \Tr\!\Big(P_{\alpha_n}\,\mathcal{U}_{n,n-1}\!\big(P_{\alpha_{n-1}} \cdots \mathcal{U}_{2,1}\!\big(P_{\alpha_1}\,\mathcal{U}_{1,0}(\rho)\,P_{\alpha_1}\big) \cdots P_{\alpha_{n-1}}\big)\Big).
\label{eq:aperture_joint}
\end{equation}
The history-conditioned post-interaction state is defined recursively by \(\sigma_0 \coloneqq \rho\) and
\begin{equation}
\sigma_k(\alpha_k, \ldots, \alpha_1) \coloneqq \mathcal{M}_{\alpha_k} \circ \mathcal{U}_{k,k-1}\!\big(\sigma_{k-1}(\alpha_{k-1}, \ldots, \alpha_1)\big).
\label{eq:state_recursion}
\end{equation}
\end{definition}

\begin{definition}[Markovian Aperture Record]
\label{def:markov_record}
The Aperture record process is \textbf{Markovian} if for each \(k\) and each \(\beta \in \{1, 2, 3\}\), the conditional probability \(\mathrm{Pr}(\alpha_{k+1} = \beta \mid \alpha_k, \ldots, \alpha_1)\) depends only on \(\alpha_k\).
\end{definition}

\begin{definition}[Aperture Sufficiency (Reset Property)]
\label{def:reset_property}
The Aperture interaction is \textbf{sufficient} if for each \(k\) and each \(\alpha \in \{1, 2, 3\}\) there exists a state \(\omega_\alpha^{(k)}\) supported on \(P_\alpha\,\mathcal{H}\) such that for all histories ending in \(\alpha\),
\begin{equation}
\sigma_k(\alpha, \alpha_{k-1}, \ldots, \alpha_1) = \omega_\alpha^{(k)}.
\label{eq:reset}
\end{equation}
\end{definition}

\begin{theorem}[Sufficient Condition for Markovianity]
\label{thm:markov_criterion}
If the Aperture interaction is sufficient (Definition~\ref{def:reset_property}), then the Aperture record process is Markovian (Definition~\ref{def:markov_record}).
\end{theorem}

\begin{proof}
By Definition~\ref{def:aperture_record},
\begin{equation}
\mathrm{Pr}(\alpha_{k+1} = \beta \mid \alpha_k, \ldots, \alpha_1) = \Tr\!\Big(P_\beta\,\mathcal{U}_{k+1,k}\!\big(\sigma_k(\alpha_k, \ldots, \alpha_1)\big)\Big).
\label{eq:conditional_prob}
\end{equation}
If the Aperture interaction is sufficient, then \(\sigma_k\) depends only on \(\alpha_k\) by \eqref{eq:reset}, so the right-hand side depends only on \(\alpha_k\).
\end{proof}

\begin{definition}[Aperture-Mixing Dynamics]
\label{def:aperture_mixing}
The AT dynamics is \textbf{Aperture-mixing} if there exist sector labels \(\alpha, \beta\) and two density operators \(\sigma \neq \sigma'\) supported on \(P_\alpha\,\mathcal{H}\) such that for some inter-interaction unitary \(U\),
\begin{equation}
\Tr\!\big(P_\beta\,U\,\sigma\,U^\dagger\big) \neq \Tr\!\big(P_\beta\,U\,\sigma'\,U^\dagger\big).
\end{equation}
\end{definition}

\begin{lemma}[Algebraic Characterization of Aperture Mixing]
\label{lem:mixing_characterization}
For sector labels \(\alpha, \beta\) and a unitary \(U\), define the \textbf{transition effect operator}
\begin{equation}
E_\beta^{(\alpha)} \coloneqq P_\alpha\,U^\dagger\,P_\beta\,U\,P_\alpha,
\end{equation}
acting on \(P_\alpha\,\mathcal{H}\).  Then for any density operator \(\sigma\) supported on \(P_\alpha\,\mathcal{H}\),
\begin{equation}
\Tr(P_\beta\,U\,\sigma\,U^\dagger) = \Tr(E_\beta^{(\alpha)}\,\sigma).
\end{equation}
The dynamics is Aperture-mixing for the pair \((\alpha, \beta)\) if and only if \(E_\beta^{(\alpha)}\) is not proportional to \(P_\alpha\).  The next-step transition probabilities from sector \(\alpha\) are independent of the hidden post-interaction state if and only if
\begin{equation}
P_\alpha\,U^\dagger\,P_\beta\,U\,P_\alpha = c_{\alpha\beta}\,P_\alpha \qquad \text{for all } \beta,
\label{eq:exact_reset}
\end{equation}
for some constants \(c_{\alpha\beta} \geq 0\) with \(\sum_\beta c_{\alpha\beta} = 1\).
\end{lemma}

\begin{proof}
By the cyclic property of the trace and \(P_\alpha\,\sigma = \sigma\,P_\alpha = \sigma\):
\begin{equation}
\Tr(P_\beta\,U\,\sigma\,U^\dagger) = \Tr(U^\dagger\,P_\beta\,U\,\sigma) = \Tr(P_\alpha\,U^\dagger\,P_\beta\,U\,P_\alpha\,\sigma) = \Tr(E_\beta^{(\alpha)}\,\sigma).
\end{equation}
If \(E_\beta^{(\alpha)} = c_{\alpha\beta}\,P_\alpha\), then \(\Tr(E_\beta^{(\alpha)}\,\sigma) = c_{\alpha\beta}\) for all normalized \(\sigma\) on \(P_\alpha\,\mathcal{H}\), so transition probabilities are state-independent.  Conversely, if \(E_\beta^{(\alpha)}\) is not proportional to \(P_\alpha\), it has at least two distinct eigenvalues on \(P_\alpha\,\mathcal{H}\); choosing \(\sigma, \sigma'\) as eigenstates with different eigenvalues yields \(\Tr(E_\beta^{(\alpha)}\,\sigma) \neq \Tr(E_\beta^{(\alpha)}\,\sigma')\).  The normalization \(\sum_\beta c_{\alpha\beta} = 1\) follows from \(\sum_\beta P_\beta = \mathbf{I}\).
\end{proof}

\begin{proposition}[Non-Markovianity Witness Criterion]
\label{prop:mixing_implies_nonmarkov}
Fix a time step \(k\), a sector label \(\alpha\), and two histories \(h,h'\) ending in \(\alpha_k=\alpha\), with post-interaction states \(\sigma_k(h) \neq \sigma_k(h')\) supported on \(P_\alpha\mathcal{H}\).  If there exists a sector label \(\beta\) such that
\begin{equation}
\Tr\!\big(E_\beta^{(\alpha)}\,\sigma_k(h)\big)\neq \Tr\!\big(E_\beta^{(\alpha)}\,\sigma_k(h')\big),
\label{eq:witness_hypothesis}
\end{equation}
then the Aperture record process is non-Markovian.
\end{proposition}

\begin{proof}
By \eqref{eq:conditional_prob} and Lemma~\ref{lem:mixing_characterization},
\[
\Pr(\alpha_{k+1}=\beta \mid h) = \Tr\!\big(E_\beta^{(\alpha)}\,\sigma_k(h)\big)
\quad\text{and}\quad
\Pr(\alpha_{k+1}=\beta \mid h') = \Tr\!\big(E_\beta^{(\alpha)}\,\sigma_k(h')\big).
\]
By \eqref{eq:witness_hypothesis} these are unequal, so the conditional probability of \(\alpha_{k+1}\) depends on information beyond \(\alpha_k\).
\end{proof}

\begin{theorem}[Non-Markovianity of the AT Aperture Record]
\label{thm:at_nonmarkovianity}
Let \(\alpha\) be a sector with internal rank \(r_\alpha \coloneqq \dim_{\C}(P_\alpha\,\mathcal{H}_F) \geq 2\).  Let \(U\) be an inter-interaction unitary for which \(E_\beta^{(\alpha)} = P_\alpha\,U^\dagger\,P_\beta\,U\,P_\alpha\) is not proportional to \(P_\alpha\) for some \(\beta\).  If two histories \(h, h'\) ending in \(\alpha_k = \alpha\) realize post-interaction states \(\sigma_k(h) \neq \sigma_k(h')\) that are distinguished by \(E_\beta^{(\alpha)}\), then the Aperture record process admits a non-Markovianity witness.
\end{theorem}

\begin{proof}
The hypotheses supply exactly the inputs to Proposition~\ref{prop:mixing_implies_nonmarkov}: two histories ending in the same sector with different post-interaction states, and a non-scalar transition effect operator that distinguishes them.
\end{proof}

\begin{remark}[Memory hierarchy]
\label{rem:nonmarkov_hierarchy}
In a one-dimensional sector (\(r_\alpha = 1\)), the transition effect operator \(E_\beta^{(\alpha)}\) is automatically scalar, so the Aperture record restricted to exits from that sector is automatically Markovian.  This statement is specific to the sequential record process with a L\"uders interaction at each step; non-Markovianity of the unmeasured transition law between configurations (Barandes~\cite{Barandes2025}) is a separate phenomenon and is not in tension with the present claim.  In sectors with \(r_\alpha \geq 2\), the operator \(E_\beta^{(\alpha)}\) may have up to \(r_\alpha\) distinct eigenvalues, permitting hidden intra-sector distinctions that affect future transition probabilities.
\end{remark}

\begin{remark}[Connection to stochastic reformulations of quantum theory]
\label{rem:stochastic_interpretation}
Theorem~\ref{thm:at_nonmarkovianity} provides a concrete structural realization of the intrinsic non-Markovian character of quantum mechanics identified by Barandes \cite{Barandes2025} in the stochastic-quantum correspondence.  In that framework, quantum systems are equivalent to fundamentally non-Markovian stochastic processes.  In AT, this non-Markovianity is derived from a specific algebraic mechanism: the coarse-graining of deterministic unitary dynamics through the abelian Aperture interface \(\mathfrak{A}_{F,b}\), whose limited tomographic capacity cannot track the internal degrees of freedom that influence future boundary statistics.  The transition effect operator \(E_\beta^{(\alpha)}\) gives a sharp algebraic criterion connecting the two perspectives.
\end{remark}

\bibliographystyle{plain}
\bibliography{references}
\end{document}